\documentclass[11pt]{article}
\usepackage{a4wide}
\usepackage{amsmath}
\usepackage[dvips]{graphics,color}
\usepackage{epsfig}
\usepackage{amssymb}
\usepackage{amsthm}
\usepackage{a4wide}
\def\re{\mathbb{R}}

\newtheorem{thm}{Proposition}
\newtheorem{dfn}{Definition}
\newtheorem{lem}{Lemma}

\begin{document}

\title{Missing observation analysis for matrix-variate time series data}

\author{K. Triantafyllopoulos\footnote{Department of Probability and Statistics, University of Sheffield, Sheffield S3 7RH,
UK, email: {\tt k.triantafyllopoulos@sheffield.ac.uk}}}

\date{\today}

\maketitle

\begin{abstract}

Bayesian inference is developed for matrix-variate dynamic linear
models (MV-DLMs), in order to allow missing observation analysis,
of any sub-vector or sub-matrix of the observation time series
matrix. We propose modifications of the inverted Wishart and
matrix $t$ distributions, replacing the scalar degrees of freedom
by a diagonal matrix of degrees of freedom. The MV-DLM is then
re-defined and modifications of the updating algorithm for missing
observations are suggested.

\textit{Some key words:} Bayesian forecasting, dynamic models,
inverted Wishart distribution, state space models.
\end{abstract}

\section{Introduction}

Suppose that, in the notation of West and Harrison (1997, Chapter
16), the $p\times r$ matrix-variate time series $\{y_t\}$ follows
a matrix-variate dynamic linear model (MV-DLM) so that
\begin{equation}\label{model}
y_t'=F_t'\Theta_t+\epsilon_t' \quad \textrm{and} \quad \Theta_t =
G_t \Theta_{t-1}+\omega_t,
\end{equation}
where $F_t$ is a $d\times r$ design matrix, $G_t$ is a $d\times d$
evolution matrix and $\Theta_t$ is a $d\times p$ state matrix.
Conditional on a $p\times p$ covariance matrix $\Sigma$, the
innovations $\epsilon_t$ and $\omega_t$ follow, respectively,
matrix-variate normal distributions (Dawid, 1981), i.e.
$$
\epsilon_t|\Sigma\sim N_{r\times p} (0,V_t,\Sigma)\quad \textrm{and}
\quad \omega_t|\Sigma\sim N_{d\times p} (0,W_t,\Sigma).
$$
This is equivalent to writing $\textrm{vec}(\epsilon_t)|\Sigma\sim
N_{rp}(0,\Sigma\otimes V_t)$ and
$\textrm{vec}(\omega_t)|\Sigma\sim N_{dp}(0,\Sigma\otimes W_t)$,
where $\textrm{vec}(.)$ denotes the column stacking operator of a
matrix, $\otimes$ denotes the Kronecker product of two matrices
and $N_{rp}(.,.)$ denotes the multivariate normal distribution.

We assume that the innovation series $\{\epsilon_t\}$ and
$\{\omega_t\}$ are internally and mutually uncorrelated and also
they are uncorrelated with the assumed initial priors
\begin{equation}\label{priors1}
\Theta_0|\Sigma\sim N_{d\times
p}(m_0,P_0,\Sigma)\quad\textrm{and}\quad \Sigma \sim IW_p
(n_0,n_0S_0),
\end{equation}
for some known $m_0$, $P_0$, $n_0$ and $S_0$. Here $\Sigma\sim
IW_p(n_0,n_0S_0)$ denotes the inverted Wishart distribution with
$n_0$ degrees of freedom and parameter matrix $n_0S_0$. The
covariance matrices $V_t$ and $W_t$ are assumed known; usually
$V_t=I_r$ (the $r\times r$ identity matrix) and $W_t$ can be
specified using discount factors as in West and Harrison (1997,
Chapter 6). Alternatively, $W_t=W$ may be considered
time-invariant and it can be estimated from the data using the EM
algorithm (Dempster {\it et al.}, 1977; Shumway and Stoffer,
1982). With the above initial priors (\ref{priors1}) the posterior
distribution of $\Theta_t|\Sigma,y_1,\ldots,y_t$ is a
matrix-variate normal distribution and the posterior distribution
of $\Sigma|y_1,\ldots,y_t$ is an inverted Wishart distribution
with degrees of freedom $n_t=n_{t-1}+1$ and a parameter matrix
$n_tS_t$, which are calculated recurrently (West and Harrison,
1997, Chapter 16).

Missing data in time series are typically handled by evaluating
the likelihood function (Jones, 1980; Ljung, 1982; Shumway and
Stoffer, 1982; Harvey and Pierse, 1984; Wincek and Reinsel, 1984;
Kohn and Ansley, 1986; Ljung, 1993; G\'{o}mez and Maravall, 1994;
Luce\~{n}o, 1994; Luce\~{n}o, 1997). In the context of model
(\ref{model}) a major obstacle in inference is when a sub-vector
or sub-matrix $\widetilde{y}_t$ of $y_t$ is missing at time $t$.
Then the scalar degrees of freedom of the inverted Wishart
distribution of $\Sigma|y_1,\ldots,y_t$, are incapable to include
the information of the observed part of $y_t$, but to exclude the
influence of the missing part $\widetilde{y}_t$. For example
consider $p=2$ and $r=1$ or $y_t=[y_{1t}~y_{2t}]'$ and suppose
that at time $t$, $y_{1t}$ is missing ($\widetilde{y}_t=y_{1t}$),
while $y_{2t}$ is observed. Let $n_{t-1}$ denote the degrees of
freedom of the inverted Wishart distribution of
$\Sigma|y_1,\ldots,y_{t-1}$. One question is how one should update
$n_t$, since the information at time $t$ is partial (one component
observed and one missing). Likewise, given this partial
information at time $t$, another question is how to estimate the
off-diagonal elements of $\Sigma$, which leads to the estimation
of the covariance of $y_{1t}$ and $y_{2t}$.

In this paper, introducing several degrees of freedom that form a
diagonal matrix, we propose modifications to the inverted Wishart
and matrix $t$ distributions. We prove the conjugacy between these
distributions and we discuss modifications in the recursions of the
posterior moments in the presence of missing data. This approach
does not require to order all missing observations in one matrix
(Shumway and Stoffer, 1982; Luce\~{n}o, 1997) and therefore it can
be applied for sequential purposes as new data are observed.

\section{Matrix-variate dynamic linear models}\label{mvdlm}

\subsection{Modified inverted Wishart distribution}\label{c4s2}

Suppose that $\Sigma$ is a $p\times p$ random covariance matrix,
$S,R$ are $p\times p$ covariance matrices and $N$ is a $p\times p$
diagonal matrix with positive diagonal elements. Let tr$(.)$,
etr$(.)$ and $|.|$ denote the trace, the exponent of the trace and
the determinant of a square matrix, respectively. The density of
the inverted Wishart distribution is given by
\begin{equation}\label{eq00}
p(\Sigma)=c|R|^{(k-p-1)/2} |\Sigma|^{-k/2}\textrm{etr}\left(
-\frac{1}{2} R\Sigma^{-1} \right),
\end{equation}
from which it is deduced that
\begin{equation}\label{eq0}
\int_{\Omega} |\Sigma|^{-k/2}\textrm{etr}\left( -\frac{1}{2}
R\Sigma^{-1} \right)\,d\Sigma=c^{-1}|R|^{-(k-p-1)/2},
\end{equation}
with $\Omega=\{\Sigma\in \re^{p\times p}:\Sigma>0\}$,
$c^{-1}=2^{(k-p-1)p/2} \Gamma_p \{(k-p-1)/2\}$, and $k>2p$, where
$\Gamma_p(.)$ is the multivariate gamma function.

\begin{lem}
The function
\begin{equation}\label{eq1}
p(\Sigma)=c |\Sigma |^{-\left\{ v+ \textrm{tr}
\left(N\right)/(2p)\right\} } \textrm{etr} \left( -\frac{1}{2}
N^{1/2} SN^{1/2}\Sigma^{-1}  \right),
\end{equation}
where $c$ does not depend on $\Sigma$, is a density function.
\end{lem}

\begin{proof}
If the following bijective transformation is applied
\begin{equation}\label{tr1}
R=N^{1/2}SN^{1/2}\quad \mbox{and} \quad
k=2v+\frac{\textrm{tr}(N)}{p},
\end{equation}
then (\ref{eq1}) is directly obtained from (\ref{eq00}).
\end{proof}

From the above bijection and the Wishart integral, we can see that
the normalizing constant $c$ is
\begin{equation*} c=c_0|S|^{\left\{
2v+\textrm{tr}\left(N\right)/p-p-1\right\}/2} \left( \prod_{j=1}^p
n_j\right)^{\left\{2v+\textrm{tr}\left(N\right)/p-p-1\right\}/2},
\end{equation*}
where
\begin{equation*}
c_0^{-1}=2^{\left\{2v+\textrm{tr}\left(N\right)/p-p-1\right\}p/2}
\Gamma_p \left\{
\frac{2v+\textrm{tr}\left(N\right)/p-p-1}{2}\right\},
\end{equation*}
for $N=\textrm{diag}(n_1,\ldots ,n_p)$ and $n_i>0$
$(i=1,\ldots,p)$.

Density (\ref{eq1}) proposes a modification of the inverted
Wishart distribution in order to incorporate a diagonal matrix of
degrees of freedom. The modification consists of a bijective
transform of the two distributions. We will then say that $\Sigma$
follows the {\it modified inverted Wishart} distribution and we
will write $\Sigma\sim MIW_p(S,N,v)$, where $v$ is a scalar
hyperparameter. Note that when $n_1=\cdots =n_p=n$ and $v=p$, the
above distribution reduces to an inverted Wishart distribution
with $n$ degrees of freedom.

With $k$ and $R$ as defined in equation (\ref{tr1}), the mean of
$\Sigma$ is
\begin{equation*}
E(\Sigma)=\frac{R}{k-2p-2}=\left\{\frac{\textrm{tr}\left(N\right)}{p}+2v-2p-2
\right\}^{-1}N^{1/2}SN^{1/2},
\end{equation*}
for $p^{-1}\textrm{tr}(N)>2p-2v+2$. The next result gives the
distribution of a $MIW$ matrix conditional on a normal matrix.
\begin{thm}\label{th21}
Let $Y$ be an $r\times p$ random matrix that follows a matrix normal
distribution, conditional on $\Sigma$, and $\Sigma$ a $p\times p$
covariance random matrix that follows a modified inverted Wishart
distribution, written $Y|\Sigma\sim N_{r\times p}(m,P,\Sigma)$ and
$\Sigma\sim MIW_p(S,N,v)$ respectively, for some known quantities
$m$, $P$, $S$, $N$, and $v$. Then, the conditional distribution of
$\Sigma$ given $Y$, is
\begin{equation*}
\Sigma|Y\sim MIW_p(S^*,N^*,v),
\end{equation*}
where $N^{*1/2}S^{*}N^{*1/2}=(Y-m)'P^{-1}(Y-m) +N^{1/2}SN^{1/2}$ and
$N^{*}=N+rI_p$.
\end{thm}

\begin{proof}
Form the joint distribution of $Y$ and $\Sigma$ and write
\begin{eqnarray}
p(\Sigma|Y)&\propto &p(Y,\Sigma)=p(Y|\Sigma)p(\Sigma)\nonumber\\
&\propto &|\Sigma|^{-\left\{ v+r/2+\textrm{tr}(N)/(2p) \right\}
}\textrm{etr}\bigg[-\frac{1}{2}
 \{ (Y-f)'Q^{-1}(Y-f)\nonumber\\
&&+N^{1/2} SN^{1/2}\}\Sigma^{-1} \bigg],\label{eq:app:joint}
\end{eqnarray}
which is sufficient for the proof with the definition of $S^{*}$ and
$N^{*}$.
\end{proof}

In the context of Proposition \ref{th21} the joint distribution of
$Y$ and $\Sigma$ is referred to as joint normal modified inverted
Wishart distribution with notation $Y,\Sigma\sim NMIW_{r\times
p,p}(m,P,S,N,v)$, for $m$, $P$, $S$, $N$, and $v$ as defined in
Proposition \ref{th21}. The next result gives the marginal
distribution of $Y$. First we give some background material on the
matrix $t$ distribution.

Let $X$ be an $r\times p$ random matrix. Then, the matrix $t$
distribution is defined by
\begin{equation}\label{eqa6}
p(X)=c|Q+(X-M)'P^{-1}(X-M)|^{-(k+r+p-1)/2},
\end{equation}
with
\begin{displaymath}
c=\frac{\Gamma_p\{(k+r+p-1)/2\}|Q|^{(k+p-1)/2}|P|^{-p/2}}{\pi^{rp/2}\Gamma_p\{(k+p-1)/2\}},
\end{displaymath}
where $M$ is an $r\times p$ matrix, $P$ a $r\times r$ covariance
matrix, $Q$ a $p\times p$ covariance matrix, and $k$ any positive
real number.

\begin{thm}\label{th3}
Let $Y$ be an $r\times p$ random matrix that follows a matrix
normal distribution conditional on $\Sigma$, and $\Sigma$ be a
$p\times p$ covariance random matrix that follows a modified
inverted Wishart distribution, written $Y|\Sigma\sim N_{r\times
p}(f,Q,\Sigma)$, and $\Sigma\sim MIW_p(S,N,v)$ respectively, for
known quantities $f$, $Q$, $S$, $N$, and $v$. Then, the marginal
distribution of $Y$ is
\begin{equation}\label{eq81}
p(Y)=c |N^{1/2}SN^{1/2}+(Y-f)'Q^{-1}(Y-f)|
^{-\left\{2v+\textrm{tr}\left(N\right)/p+d-p-1\right\}/2},
\end{equation}
which by analogy of the $MIW$ distribution, is a modification of the
matrix $t$ distribution and it is written as $MT(f,Q,S,N,v)$.
\end{thm}

\begin{proof}
The joint distribution of $Y$ and $\Sigma$ is given by equation
(\ref{eq:app:joint}). Hence, the marginal distribution of $Y$ is
\begin{displaymath}
p(Y)=\int_{\Omega}p(Y,\Sigma)\,d\Sigma,
\end{displaymath}
where $\Omega=\{\Sigma\in \re^{p\times p}:\Sigma>0\}$. Set
$R=(Y-f)'Q^{-1}(Y-f)+N^{1/2} SN^{1/2}$ and $k=2v+r+\textrm{tr}(N)/p$
and from equation (\ref{eq0}) we have equation (\ref{eq81}).
\end{proof}

The distribution of Proposition (\ref{th3}) can be derived from the
matrix $t$ distribution (see equation (\ref{eqa6})). The normalizing
constant $c$ of (\ref{eq81}) is obtainable from (\ref{eqa6}) as
\begin{equation*}
c=\frac{\pi^{pr/2}\Gamma_p\{ (k+p-1)/2\} }{\Gamma_p\{ (k+r+p-1)/2\}
}|S|^{(k+p-1)/2}\left(\prod_{j=1}^p
n_j\right)^{(k+p-1)/2}|Q|^{-p/2},
\end{equation*}
where $N=\textrm{diag}(n_1,\ldots,n_p)$ and
$k=2v-2p+\textrm{tr}(N)/p$. Note that if all the diagonal elements
of $N$ are the same (i.e. $n_1=\cdots =n_p=n$) and $v=p$, then the
above distribution reduces to a matrix $t$ distribution with $n$
degrees of freedom.

Finally we give the marginal distribution of $\Sigma$. Consider
the following partition of $\Sigma$, $S$, and $N$
\begin{displaymath}
\Sigma=\left[\begin{array}{cc} \ \Sigma_{11} & \Sigma_{12}
\\ \Sigma_{12}' & \Sigma_{22}
\end{array}\right], \quad S=\left[\begin{array}{cc} \ S_{11} &
S_{12} \\ S_{12}' & S_{22}
\end{array}\right], \quad N=\left[\begin{array}{cc} \ N_1 &
0 \\ 0' & N_2
\end{array}\right],
\end{displaymath}
where $\Sigma_{11}$, $S_{11}$ and $N_{11}$ have dimension $q\times
q$, for some $1\leq q<p$. The next result gives the marginal
distribution of $\Sigma_{11}$.
\begin{thm}\label{th2}
If $\Sigma\sim MIW_p(S,N,v)$, under the above partition of
$\Sigma$ the distribution of $\Sigma_{11}$ is $\Sigma_{11}\sim
MIW_q(S_{11},N_{11},v_1)$, where
$v_1=v-p+q+2^{-1}p^{-1}\textrm{tr}(N)-2^{-1}q^{-1}\textrm{tr}(N_1)$.
\end{thm}

\begin{proof}
The proof suggests the adoption of transformation (\ref{tr1})
together with the partition of $R$ in (\ref{eq00}) as
\begin{displaymath}
R=\left[\begin{array}{cc} R_{11} & R_{12}\\
R_{12}' & R_{22}\end{array}\right].
\end{displaymath}
Using marginalization properties of the inverted Wishart
distribution, upon noticing
\begin{displaymath}
N^{1/2}SN^{1/2}=\left[\begin{array}{cc} \ N_1^{1/2}S_{11}N_1^{1/2}
&
N_1^{1/2}S_{12}N_2^{1/2}\\
N_2^{1/2}S_{12}'N_1^{1/2} &
N_2^{1/2}S_{22}N_2^{1/2}\end{array}\right],
\end{displaymath}
we get $\Sigma_{11}\sim MIW_q(S_{11},N_1,v_1)$, with $v_1$ as
required.
\end{proof}

A similar result can be obtained for $\Sigma_{22}$. Consequently,
if we write $\Sigma=\{\sigma_{ij}\}$ $(1\leq i,j\leq p)$ and
$N=\textrm{diag}(n_1,\ldots,n_p)$, then the diagonal variances
$\sigma_{ii}$ follow modified inverted Wishart distributions,
$\sigma_{ii}\sim MIW_1(s_{ii},n_i,v_i)$, where
$v_i=v-p+1+2^{-1}p^{-1}\textrm{tr}(N)-2^{-1}n_i$. These in fact
are inverted gamma distributions $\sigma_{ii}\sim
IG(v_i+n_i/2-1,n_is_{ii}/2)$. Note that if $n_1=\cdots =n_p=n$ and
$v=p$, then we have that $\sigma_{ii}\sim IG(n/2,ns_{ii}/2)$ (the
inverted gamma distribution used in West and Harrison (1997) when
$p=1$).

We close this section with a brief discussion on an earlier study
proposing the incorporation of several degrees of freedom for
inverted Wishart matrices (Brown {\it et al.}, 1994). This
approach is based on breaking the degrees of freedom on blocks and
requiring for each block the marginal density of the covariance
matrix to follow an inverted Wishart distribution. However, in
that framework the conjugacy between the normal and that
distribution is lost and as a result the proposed estimation
procedure may be slow and probably not suitable for time series
application. Relevant inferential issues of that approach are
discussed in Garthwaite and Al-Awadhi (2001). Our proposal of the
$MIW$ distribution retains the desired conjugacy and it leads to
relevant modifications of the matrix $t$ distribution, which
provides the forecast distribution. Furthermore, the $MIW$ density
leads to fast computationally efficient algorithms, which are
suitable for sequential model monitoring and expert intervention
(Salvador and Gargallo, 2004). Finally, according to Proposition
\ref{th2}, the marginal distributions of $MIW$ matrices are also
$MIW$, which means that several degrees of freedom are included in
the marginal models too, something that is not the case in the
approach of Brown {\it et al.} (1994).

\subsection{Matrix-variate dynamic linear models
revisited}\label{c4s3}

We consider model (\ref{model}), but now we replace the initial
priors (\ref{priors1}) by the priors
\begin{equation}\label{priors}
\Theta_0|\Sigma\sim N_{d\times
p}(m_0,P_0,\Sigma_0)\quad\textrm{and}\quad \Sigma_0 \sim MIW_p
(S_0,N_0,p),
\end{equation}
for some known $m_0$, $P_0$, $S_0$ and $N_0$. Practically we have
replaced the inverted Wishart prior by the $MIW$ and so, for each
$t=1,\ldots,T$, we use $p$ degrees of freedom $n_{1t},\ldots,n_{pt}$
in order to estimate $\Sigma|y^t$, where $y^t$ denotes the
information set, comprising of observed data $y_1,\ldots,y_t$. The
next result provides the posterior and forecast distributions of the
new MV-DLM.

\begin{thm}\label{th4}
One-step forecast and posterior distributions in the model
(\ref{model}) with the initial priors (\ref{priors}), are given,
for each $t$, as follows.
\begin{description}
\item (a) Posterior at $t-1:\qquad \Theta_{t-1},\Sigma|y^{t-1}\sim
NMIW_{d\times p,p} (m_{t-1},P_{t-1},S_{t-1},N_{t-1},p)$,\\
for some $m_{t-1}$, $P_{t-1}$, $S_{t-1}$ and $N_{t-1}$.
\item (b) Prior at $t:\qquad\Theta_t,\Sigma|y^{t-1}\sim
NMIW_{d\times p,p} (a_t,R_t,S_{t-1},N_{t-1},p)$,\\
where $a_t=G_tm_{t-1}$ and $R_t=G_tP_{t-1}G_t'+W_t$.
\item (c) One-step forecast at $t$:$\qquad y_t'|\Sigma,y^{t-1}\sim
N_{r\times p} (f_t',Q_t,\Sigma)$,\\
with marginal:$\qquad y_t'|y^{t-1}\sim
MT_{r\times p} (f_t',Q_t,S_{t-1},N_{t-1},p)$,\\
where $f_t'=F_t'a_t$ and $Q_t=F_t'R_tF_t+V_t$.
\item (d) Posterior at $t:\qquad \Theta_t,\Sigma|y^t\sim
NMIW_{d\times p,p} (m_t,P_t,S_t,N_t,p)$,\\
with
\begin{gather*}
m_t=a_t+A_te_t',\quad
P_t=R_t-A_tQ_tA_t',\\
N_t=N_{t-1}+rI_p, \quad N_t^{1/2}S_tN_t^{1/2}=N_{t-1}^{1/2}S_{t-1}
N_{t-1}^{1/2}+e_tQ_t^{-1}e_t',\\
A_t=R_tF_tQ_t^{-1}, \quad \mbox{and} \quad e_t=y_t-f_t.
\end{gather*}
\end{description}
\end{thm}
The proof of this result follows immediately from Propositions
\ref{th21} and \ref{th3}. For $t=1$, (a) coincides with the priors
(\ref{priors}). From Proposition \ref{th3}, the marginal posterior
of $\Theta_t|y^t$ is $\Theta_t|y^t\sim MT_{d\times
p}(m_t,P_t,S_t,N_t,p)$. Thus the above proposition gives a
recursive algorithm for the estimation and forecasting of the
system for all $t=1,\ldots,T$.

Proposition \ref{th4} gives a generalization of the updating
recursions of matrix-variate dynamic models (West and Harrison,
1997, Chapter 16). The main difference of the two algorithms is that
the scalar degrees of freedom $n_t$ of the standard recursions are
replaced by $N_t$ in the above proposition and that the inverted
Wishart distribution is replaced by the modified inverted Wishart
distribution (in order to account for the matrix of degrees of
freedom). As a result the classical Bayesian updating of West and
Harrison (1997) is obtained as a special case of the distributional
results of Proposition \ref{th4}, by setting
$N_t=n_tI_p=\textrm{diag}(n_t,\ldots,n_t)$ $(t=0,1,\ldots,T)$, where
$n_t$ represent the scalar degrees of freedom of the inverted
Wishart distribution of $\Sigma_t|y^t$ and $n_0$ is the initial
degrees of freedom.

\section{Missing observations}\label{c4s4}

In this section we consider missing observations at random. Our
approach is based on excluding any missing values of the calculation
of the updating equations (state and forecast distributions) thus
excluding the unknown influence of these unobserved variables. This
approach is explained for univariate dynamic models in West and
Harrison (1997, Chapters 4,10).

The univariate dynamic linear model with unknown observational
variance is obtained from model (\ref{model}) for $p=r=1$. In this
case the posterior recursions of $m_t$, $P_t$ and $S_t$ of West
and Harrison (1997, Chapter 4) follow from Proposition \ref{th4}
as a special case. Now suppose that at time $t$ the scalar
observation $y_t$ is missing so that $y^t=y^{t-1}$. It is then
obvious that the posterior distribution of $\Theta_t$ equals its
prior distribution (since no information comes in to the system at
time $t$). Then we have $m_t=a_t$, $P_t=R_t$, $S_t=S_{t-1}$ and
$N_t=n_t=n_{t-1}=N_{t-1}$. To incorporate this into the updating
equations of the posterior means and variances, we can write
$m_t=a_t-A_te_tu_t$, $P_t=R_t-A_tA_t'Q_tu_t$,
$n_tS_t=n_{t-1}S_{t-1}+e_t^2u_t/Q_t$ and $n_t=n_{t-1}+u_t$, where
$u_t$ is zero, if $y_t$ is missing and $u_t=1$, if $y_t$ is
observed. So when $p=1$ the inclusion of $u_t$ in the posterior
recursions leads to identical analysis as in West and Harrison
(1997) and in references therein. The introduction of $u_t$ in the
recursions automates the posterior/prior updating in the presence
of missing values and it motivates the case for $p,r\geq 1$.

Moving to the multivariate case, first we consider model
(\ref{model}) as defined in the previous section with $r=1$.
Assume that we observe all the $p\times 1$ vectors $y_{i}$,
$i=1,\ldots ,t-1$. At time $t$ some observations are missing
(sub-vectors of $y_t$, or the entire $y_t$). To distinguish the
former from the latter case we have the following definition.
\begin{dfn}\label{df41}
A partial missing observation vector is said to be any strictly
sub-vector of the observation vector that is missing. If the entire
observation vector is missing it is referred to as full missing
observation vector.
\end{dfn}
Considering the MV-DLM (\ref{model}), it is clear that in the case
of a full missing vector we have
\begin{equation}\label{eq20}
\Theta_t,\Sigma|y^t\sim NMIW_{d\times p,p} (m_t,P_t,S_t, N_t,p),
\end{equation}
where $m_t=a_t$, $P_t=R_t$, $S_t=S_{t-1}$, $N_t=N_{t-1}$, since no
information comes in at time $t$. This equation relates to the
standard posterior distribution of West and Harrison (1997) by
setting $N_t=\textrm{diag}(n_t,\ldots,n_t)$, for a scalar $n_t>0$
and evidently reducing the $MIW$ distribution by a $IW$
distribution. If one starts with a prior
$N_0=\textrm{diag}(n_0,\ldots,n_0)$, and assuming that at some
time $t$, there is a full missing vector $y_t$, then it is clear
that the posterior (\ref{eq20}) equals to the posterior of
$\Theta_t,\Sigma|y^t$ using the standard recursions (West and
Harrison, 1997). Any differences between the two algorithms is
highlighted only by observing partial missing vectors and this has
been the motivation of the new algorithm.

Define a $p\times p$ diagonal matrix
$U_t=\mbox{diag}(i_{1t},\ldots ,i_{pt})$ with
\begin{displaymath}
i_{jt}=\left\{ \begin{array}{cc} 1 & \textrm{if $y_{jt}$ is
observed},\\ 0 & \textrm{if $y_{jt}$ is missing,}
\end{array}\right.
\end{displaymath}
for all $1\le j\le p$, where $y_t=[y_{1t}~\cdots ~y_{pt}]'$.

Then, the posterior distribution (\ref{eq20}) still applies with
recurrences
\begin{gather}
m_t=a_t+A_te_t'U_t\label{eq21}\\
P_t=R_t-A_tA_t'Q_tu_t\label{eq22}\\
N_t=N_{t-1}+U_t\label{eq23}\\
N_t^{1/2}S_tN_t^{1/2}=N_{t-1}^{1/2}S_{t-1}
N_{t-1}^{1/2}+U_te_tQ_t^{-1}e_t'U_t,\label{eq24}
\end{gather}
where $u_t=\mbox{tr}(U_t)/p$. Some explanation for the above
formulae are in order.

First note that if no missing observation occurs $U_t=I_p$, $u_t=1$
and we have the standard recurrences as in Proposition \ref{th4}. On
the other extreme (full missing vector), $U_t=0$, $u_t=0$ and we
have equation (\ref{eq20}). Consider now the case of partial missing
observations. Equation (\ref{eq23}) is the natural extension of the
single degrees of freedom updating, see West and Harrison (1997,
Chapter 16). For equation (\ref{eq21}) note that the zero's of the
main diagonal of $U_t$ convey the idea that the corresponding to the
missing values elements of $m_t$ remain unchanged and equal to
$a_t$. For example, consider the case of $p=2$, $d=2$ and assume
that you observe $y_{1t}$, but $y_{2t}$ is missing. Then
\begin{displaymath}
m_t=a_t+\left[\begin{array}{cc} \ A_{1t}(y_{1t}-f_{1t}) & 0\\
A_{2t}(y_{1t}-f_{1t}) & 0\end{array}\right],
\end{displaymath}
where $A_t=[A_{1t}~A_{2t}]'$. The zero's on the right hand side
reveal that the second column of $m_t$ is the same as the second
column of $a_t$. Similar comments apply for equations (\ref{eq22})
and (\ref{eq24}).

Considering the case of $r\geq 2$, we define $U_{kt}$ to be the
diagonal matrix $U_{kt}=\mbox{diag}(i_{1k,t},\ldots,i_{pk,t})$ with
\begin{displaymath}
i_{jk,t}=\left\{ \begin{array}{cc} 1 & \textrm{if $y_{jk,t}$ is
observed},\\ 0 & \textrm{if $y_{jk,t}$ is missing,}
\end{array}\right.
\end{displaymath}
where $y_t=\{y_{jk,t}\}$, $(j=1,\ldots,p;k=1,\ldots,r)$.

Then, the moments of equation (\ref{eq20}) can be updated via
\begin{gather*}
m_t=a_t+A_te_t'\prod_{k=1}^rU_{kt}, \quad P_t=R_t-A_tQ_tA_t'u_t,
\quad
N_t=N_{t-1}+\sum_{k=1}^rU_{kt}\\
N_t^{1/2}S_tN_t^{1/2}=N_{t-1}^{1/2}S_{t-1}
N_{t-1}^{1/2}+\left(\prod_{k=1}^rU_{kt}\right)e_tQ_t^{-1}
e_t'\left(\prod_{k=1}^rU_{kt}\right),
\end{gather*}
where $u_t=\mbox{tr}(\prod_{k=1}^rU_{kt})/p$. Similar comments as in
the case of $r=1$ apply. Definition \ref{df41} is trivially extended
in the case when observations form a matrix ($r\geq 2$).

We illustrate the proposed methodology by considering simulated
data, consisting of 100 bivariate time series
$y_1,\ldots,y_{100}$, generated from a local level model
$y_t=[y_{1t}~y_{2t}]'=\psi_t+\epsilon_t$ and
$\psi_t=\psi_{t-1}+\zeta_t$, where $\psi_0$, $\epsilon_t$ and
$\zeta_t$ are all simulated from bivariate normal distributions.
The correlation of $\epsilon_{1t}$ and $\epsilon_{2t}$ is set to
$0.8$, while the elements of $\zeta_t$ are uncorrelated. This
model is a special case of model (\ref{model}) with
$\Theta_t'F_t=\psi_t$ and $G_t=I_2$. Figure \ref{fig1} (solid
line) shows the simulated data; the gaps in this figure indicate
missing values at times $t=24,43,60,75,86$. At times $t=24,43,86$,
$y_{t2}$ is only missing (partial missing vectors), at time
$t=75$, $y_{t1}$ is only missing (partial missing vector) and at
time $t=60$, both $y_{t1},y_{t2}$ are missing (full missing
vector). For this data set, we compare the performance of
recursions (\ref{eq21})-(\ref{eq24}) with that of the classic or
old recursions of West and Harrison (1997), which assume that when
there is at least one missing value we set $U_t=0$ and $u_t=0$.
For example using the old recursions, for $t=24$ one would set
$U_{24}=0$ and $u_{24}=0$, losing the ``partial'' information of
$y_{24,1}=-3.739$, which is observed. On the other hand, the new
recursions would suggest for $t=24$ to set
$$
U_{24}=\left[\begin{array}{cc} 1 & 0 \\ 0 & 0 \end{array}\right]
\quad \textrm{and} \quad u_{24}=1/2.
$$

\begin{figure}[t]
 \epsfig{file=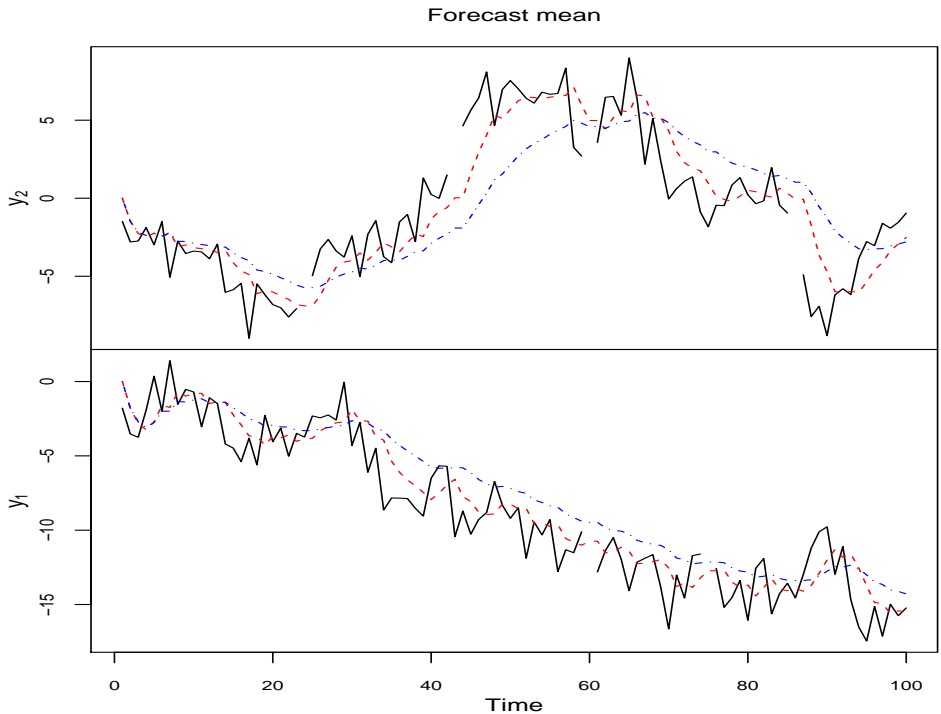, height=10cm, width=15cm}
 \caption{Simulated bivariate time series (solid line) with one-step forecasts from (a)
 the standard DLM recursions (dotted/dashed line) and (b) the new DLM recursions (dashed line). The gaps
 indicate missing values.}\label{fig1}
\end{figure}

Figure \ref{fig1} shows the one-step forecast mean of $\{y_t\}$
using the new recursions (dashed line) and the old recursions
(dotted/dashed line). We observe that the new method produces a
clear improvement in the forecasts as the old recursions provide
poor forecasts, especially in the low panel of Figure \ref{fig1}
(for $\{y_{1t}\}$). What is really happening in this case is that,
under the old recursions, the missing values of $y_{2t}$ affect
the recursions for $y_{1t}$, since the observed information at
$y_{1t}$ is wrongly ``masked'' or ``ignored'' for the points of
time when $y_{2t}$ is missing. On the other hand, the new
recursions use the explicit information from each sub-vector of
$y_t$ and thus the new recursions result in a notably more
accurate forecast performance. This is backed by the mean square
standardized forecast error vector, which for the new recursions
is $[1.300~1.825]'$, while for the old recursions is
$[1.545~2.182]'$. Under the old recursions we can not obtain an
estimate of the covariance between an observed $y_{1t}$ and a
missing $y_{2t}$. However, this is indeed obtained under the
proposed new recursions and so the respective correlations at
points of time where there are gaps are $0.633$ (at $t=24$),
$0.779$ (at $t=43$), $0.812$ (at $t=75$) and $0.809$ (at $t=86$);
the mean of these correlations is $0.792$, which is close to the
real $0.8$ under the simulation experiment.

\section*{Acknowledgements}

I am grateful to Jeff Harrison for useful discussions on the topic
of missing data in time series. I would like to thank a referee for
helpful comments.

\end{document}